\documentclass[conference]{IEEEtran}
\usepackage{epsfig,makeidx,color}
\usepackage{amsmath,mathrsfs}
\usepackage{amssymb,amsfonts}
\usepackage[english]{babel}
\usepackage[latin1]{inputenc}
\usepackage{graphicx,color}
\usepackage{amscd}
\usepackage{verbatim}
\usepackage{float}
\usepackage{dsfont}
\usepackage{citesort}
\usepackage[numbers, square, comma]{natbib}
\usepackage[ruled]{algorithm2e}
\usepackage[bookmarks=false]{hyperref}
\usepackage{balance}

  \providecommand{\norm}[1]{\ensuremath{\left\Vert #1 \right\Vert}}

\DeclareMathOperator*{\dist}{dist}  

\providecommand{\norm}[1]{\ensuremath{\left\Vert #1 \right\Vert}}

\newtheorem{lem}{Lemma}

\newtheorem{defn}{Definition}

\newtheorem{theorem}{Theorem}

\input{tcilatex}
\newcommand{\mylist}{
   \begin{list}{-}
     {
      \setlength{\topsep}{2.5ex}
      \setlength{\itemsep}{2ex}
      \setlength{\leftmargin}{2.5ex}
      } }
\newcommand{\listend}{
    \end{list}  }
\begin{document}

\title{Proximity Factors of Lattice Reduction-Aided Precoding for Multiantenna Broadcast}

\pubid{} \specialpapernotice{}
\author{ \authorblockN{Shuiyin Liu and Cong Ling}
\authorblockA{Department of Electrical and   \\       Electronic Engineering
      \\   Imperial College London       \\      London, SW7 2AZ, UK    \\   shuiyin.liu06@imperial.ac.uk,
cling@ieee.org} \and
\authorblockN{Xiaofu Wu}
\authorblockA{Nanjing University of Posts    \\       Telecommunications
      \\          \\      Nanjing, 210003, China   \\   xfuwu@ieee.org} }

\maketitle%

\begin{abstract}%

Lattice precoding is an effective strategy for multiantenna
broadcast. In this paper, we show that approximate lattice precoding
in multiantenna broadcast is a variant of the closest vector problem
(CVP) known as $\eta$-CVP. The proximity factors of lattice
reduction-aided precoding are defined, and their bounds are derived,
which measure the worst-case loss in power efficiency compared to
sphere precoding. Unlike decoding applications, this analysis does
not suffer from the boundary effect of a finite constellation, since
the underlying lattice in multiantenna broadcast is indeed infinite.


\end{abstract}%

\section{Introduction}

Broadcast is referred to as the application where a single
transmitter sends different messages to many users simultaneously.
It may arise, for example, in the downlink of a multiuser
communication system where the base station wants to communicate
with the users in the area of coverage. The multi-input multi-output
(MIMO) technology offers a new opportunity for developing efficient
broadcast strategies.

The capacity of a MIMO broadcast channel has been determined in
\cite{Weingarten06}, where it was shown that the so-called
dirty-paper coding is instrumental to achieving the capacity.
Multiple antennas allow to pre-cancel the interuser interference,
which is known at the transmitter in the broadcast application. The
lattice method represents a major approach to cancelling known
interference \cite{Erez05}, and specifically, to precoding for MIMO
broadcast. However, dirty-paper coding suffers from high complexity.

Hochwald \textit{et al.} \cite{Hochwald05} formulated precoding as a
decoding problem at the transmitter. Their technique, termed
``vector perturbation", corresponds to solving the closest lattice
vector problem (CVP). It requires the use of the sphere precoder
\cite{viterbo}, whose average complexity grows quickly with the
system size. Earlier, the idea of precoding via an algorithmic
search over modulo equivalent points was proposed by Fischer
\textit{et al.} for the intersymbyol interference channel
\cite{Fischer95}. More recently, reference \cite{Maurer11}
considered some practical issues in the implementation of vector
perturbation.

To reduce the complexity, lattice reduction (LR) can be used,
\textit{i.e.}, an approximate solution is found by zero-forcing (ZF)
or  successive interference cancelation (SIC) on a reduced lattice
\cite{Babai,Windpassinger04}. Another scheme of approximate lattice
precoding was proposed in \cite{Taherzadeh07}
(we will show that it is actually equivalent to LR-aided ZF in \cite%
{Windpassinger04}), which was shown to achieve the full diversity
order.

In contrast to the complexity analysis of sphere decoding
\cite{jalden}, the complexity of sphere precoding is not available
in literature. Moreover, the signal-to-noise (SNR) gap between
sphere precoding and LR-aided precoding has not been analyzed,
although it has been done for decoding \cite{LingIT07}. In this
paper, we investigate these aspects of lattice precoding algorithms.
We view the precoding problem as a variant of the CVP known as
$\eta$-CVP. This view enables us to derive the proximity factors for
lattice precoding, which measure the worst-case loss in power
efficiency of LR-aided precoding schemes.

The paper is organized as follows: Section II presents the model of
MIMO broadcast using lattice precoding, and investigates its
complexity. In Section III the analysis of the proximity factors is
given. Section V is a discussion.

\textit{Notation}: The transpose, inverse, pseudoinverse of a matrix $%
\mathbf{B}$ by $\mathbf{B}^{T}$, $\mathbf{B}^{-1}$, and $\mathbf{B}^{\dagger
}$, respectively, and the Euclidean length $\Vert \mathbf{u}\Vert =\sqrt{%
\langle \mathbf{u},\mathbf{u}\rangle }$. $\lceil x\rfloor $ rounds to a
closest integer.

\section{Lattice Precoding for MIMO Broadcast}

Consider a MIMO broadcast system including one transmitter, equipped
with $n$ antennas, and $n$ receivers, each equipped with a single
antenna \cite{Hochwald05}. For convenience, we use the real-valued
signal model%
\begin{equation}
\mathbf{y}=\mathbf{Hs}+\mathbf{n},
\end{equation}%
where $\mathbf{y}$ is the received signal vector at the users,
$\mathbf{H}\in \mathbb{R}^{n\times n}$ is a full-rank channel
matrix, $\mathbf{s}$ is the transmitted signal, and $\mathbf{n}$ is
the noise vector. The entry $h_{i,j}$ of $\mathbf{H}$ indicates the
channel coefficient between transmit antenna $i$ and user $j$.
$\mathbf{s}$ is derived from the data vector $\mathbf{x=}\left[
x_{1},\text{...},x_{n}\right] ^{T}$. We assume that $\mathbf{x} \in
[-A/2,A/2]^n$ is taken from the intersection of a finite hypercube
and an integer lattice. The entries of $\mathbf{n}$ are i.i.d.
Gaussian with variance $\sigma ^{2}$ each.

\subsection{Lattice Preliminaries}

An $n$-dimensional \emph{lattice} in the $m$-dimensional Euclidean space $%
\mathbb{R}^{m}$ ($n\leq m$) is the set of integer linear combinations of $n$
independent vectors $\mathbf{b}_{1},\ldots ,\mathbf{b}_{n}\in \mathbb{R}^{m}$%
:%
\begin{equation*}
\mathcal{L}\left( \mathbf{B}\right) \mathbf{=}\left\{ \sum_{i=1}^{n}x_{i}%
\mathbf{b}_{i}\left\vert x_{i}\in \mathbb{Z}\text{, }i=1,\ldots n\right.
\right\} .
\end{equation*}%
The matrix $\mathbf{B=}\left[ \mathbf{b}_{1}\cdots \mathbf{b}_{n}\right] $
is a basis of the lattice $\mathcal{L}(\mathbf{B})$. In matrix form, $%
\mathcal{L}(\mathbf{B})=\left\{ \mathbf{Bx}\text{ : }\mathbf{x\in }\text{ }%
\mathbb{Z}^{n}\right\} $. For any point $\mathbf{y\in }\mathbb{R}^{m}$ and
any lattice $\mathcal{L}\left( \mathbf{B}\right) $, the distance of $\mathbf{%
y}$ to the lattice is $\dist(\mathbf{y},\mathbf{B})=\min_{x\in \mathbb{Z}^{n}}%
\norm{\mathbf{y}-\mathbf{Bx}}$. A \emph{shortest vector} of a lattice $%
\mathcal{L}\left( \mathbf{B}\right) $ is a non-zero vector in $\mathcal{L}%
\left( \mathbf{B}\right) $ with the smallest Euclidean norm. The length of
the shortest vector, often referred to as the \emph{minimum distance}, of $%
\mathcal{L}\left( \mathbf{B}\right) $ is denoted by $\lambda _{1}$.

A lattice has infinitely many bases. In general, every matrix $\mathbf{%
\widetilde{B}=BU}$ is also a basis, where $\mathbf{U}$ is an \textit{unimodular} matrix, i.e., $%
\det (\mathbf{U})=\pm 1$ and all elements of $\mathbf{U}$ are
integers. The aim of \emph{%
lattice reduction} is to find a good basis for a given lattice. In
many applications, it is advantageous to have the basis vectors as
short as possible. The celebrated LLL algorithm is the first
polynomial (average) time algorithm which finds a vector not much
longer than
the shortest nonzero vector.

Let $\mathbf{\hat{b}}_{1}$,...,$\mathbf{\hat{b}}_{n}$ be the
Gram-Schmidt vectors corresponding to a basis
$\mathbf{b}_{1}$,...,$\mathbf{b}_{n}$, where $\mathbf{%
\hat{b}}_{i}$ is the projection of $\mathbf{b}_{i}$ orthogonal to
the vector space generated by
$\mathbf{b}_{1}$,...,$\mathbf{b}_{i-1}$. These are the vectors found
by the Gram-Schmidt algorithm for orthogonalization.
Gram-Schmidt orthogonalization (GSO) is closely related to QR decomposition $%
\mathbf{B}=\mathbf{QR}$. More precisely, one has the relations $\mu
_{j,i}=r_{i,j}/r_{i,i}$ and $\mathbf{\hat{b}}_{i}=$ $r_{i,i}\cdot \mathbf{q}%
_{i}$, where $\mathbf{q}_{i}$ is the $i^{th}$ column of
$\mathbf{Q}$.

A basis $\mathbf{B}$ is LLL reduced if%
\begin{equation}
\left\vert \mu _{i,j}\right\vert \leq 1/2
\end{equation}%
for $1\leq j<i\leq n$, and%
\begin{equation}
\Vert \mathbf{\hat{b}}_{i}\Vert ^{2}\geq \left( \delta -\mu
_{i,j}^{2}\right) \Vert \mathbf{\hat{b}}_{i-1}\Vert ^{2}  \label{R-LLL-2}
\end{equation}%
for $1<i\leq n$, where $1/4<\delta \leq 1$ is a factor selected to achieve a
good quality-complexity tradeoff.

We now define a variant of the CVP.

\begin{defn}
[$\eta $-CVP)] Given a lattice $\mathcal{%
L}\left( \mathbf{B}\right) $ and a vector $\mathbf{y\in }$ $\mathbb{R}^{m}$,
find a vector $\mathbf{B\hat{x}}\in \mathcal{L}\left( \mathbf{B}\right) $
such that $\Vert \mathbf{y}-\mathbf{B\hat{x}}\Vert $ $\leq \eta \dist(%
\mathbf{y},\mathbf{B})$.
\end{defn}

\subsection{Sphere Precoding}

In this method, the transmitted signal is given by \cite{Hochwald05}%
\begin{equation}
\mathbf{s}=\mathbf{B}( \mathbf{x}-A\mathbf{\hat{l}}) =
\mathbf{Bx}\mod \mathcal{L}(A\mathbf{B}),
\end{equation}%
where $\mathbf{B}\triangleq \mathbf{H}^{-1}$, and $\mathbf{\hat{l}}$
is an integer vector, chosen to minimize the
transmission power:%
\begin{equation}
\mathbf{\hat{l}}=\arg \min_{\mathbf{l}\in \mathbb{Z}^{n}}{\Vert \mathbf{B}(%
\mathbf{x}-A\mathbf{l})\Vert ^{2}}.  \label{min_power}
\end{equation}%
Note that $\mathbf{s} \in \mathcal{V}(\mathcal{L}(A\mathbf{B}))$
(the Voronoi region). The receivers apply the modulo function each,
obtaining
\begin{equation}
\begin{split}
\mathbf{y}\mod A& =\mathbf{H}\mathbf{B}(\mathbf{x}-A\mathbf{\hat{l}})+%
\mathbf{n}\mod A \\
& =(\mathbf{x}-A\mathbf{\hat{l}})+\mathbf{n}\mod A \\
& =\mathbf{x}+\mathbf{n}\mod A.
\end{split}
\label{rcv}
\end{equation}%
Namely, the data arrive at individual users free of interuser
interference; the only effect is noise. To solve the CVP
(\ref{min_power}), the sphere precoding algorithm originally
proposed for decoding purposes was used.

It is worth pointing out several distinctions between the CVP's in
decoding and precoding:

\begin{itemize}
\item Decoding gets easier for weaker noise, while noise has no impacts on
the hardness of lattice precoding.

\item The constellation in decoding is often finite, while the lattice in
precoding is infinite. Thus, the boundary errors in decoding will
not an issue in precoding.

\item The received signal in decoding has a Gaussian distribution centered at a
lattice point, while the input to precoding is roughly uniformly
distributed on a fundamental parallelepiped.
\end{itemize}

For these reasons, sphere precoding incurs more computational
complexity than sphere decoding at the same dimension $n$. Fincke
and Pohst \cite{fincke} proposed an algorithm to enumerate the
lattice points in a sphere, running on an LLL-reduced lattice, but
their complexity estimate was loose. Kannan's algorithm
\cite{kannan87} for HKZ reduction can be used to preprocess the
lattice, giving a CVP algorithm with $n^{n+o(n)}$ complexity. Hanrot
and Stehl\'e' improved the CVP complexity analysis to $n^{n/2+o(n)}$
\cite{Stehle07}.

On the other hand, Jald\'en and Ottersten \cite{jalden} showed that
the average complexity of sphere decoding is exponential with the
dimension for any fixed SNR; the constant within the exponent,
though, does decrease with SNR, meaning lower complexity at higher
SNR. However, the encouraging results for lattice decoding do not
extend to precoding. Noise, which is crucial to the decreasing
complexity of sphere decoding, does not even arise in lattice
precoding. Since the input is largely uniformly distributed in the
fundamental parallelepiped, the worst-case bound is a sensible
measure of complexity. Moreover, the paper \cite{jalden} assumed a
finite constellation, rendering the analysis inapplicable to an
infinite lattice, which is nonetheless the case for precoding
problems.

To conclude, the worst-cast complexity of sphere precoding is
super-exponential.

\subsection{Approximate Lattice Precoding}

\subsubsection{SIC Precoding}

To obtain a fast precoder, Windpassinger et al. \cite{Windpassinger04}
approximated the CVP by using lattice reduction, i.e., the closest vector is
replaced with Babai's approximations \cite{Babai}. Let $\widetilde{\mathbf{B}%
}$ designate the reduced basis, i.e., $\widetilde{\mathbf{B}}=\mathbf{B}%
\mathbf{U}$, where $\mathbf{U}$ is a unimodular matrix. Performing the QR decomposition $\widetilde{\mathbf{B}}=%
\mathbf{QR}$, where $\mathbf{Q}$ has orthogonal columns and
$\mathbf{R}$ is an upper triangular matrix with nonnegative diagonal
elements.

Let $\mathbf{u}=\mathbf{Q^{\dagger }\widetilde{B}x/}A$. An estimate of $%
\mathbf{\hat{l}}$ is then found by the SIC procedure:%
\begin{eqnarray}
\hat{l}_{n} &=&\left\lceil u_{n}/r_{n,n}\right\rfloor ,  \notag \\
\hat{l}_{i} &=&\left\lceil \frac{u_{i}-\sum_{j=i+1}^{n}r_{i,j}\hat{l}_{j}}{%
r_{i,i}}\right\rfloor ,i=n-1,...,1.
\end{eqnarray}%
The transmitted signal is given by%
\begin{equation}
\mathbf{s}=\mathbf{Bx}-A\widetilde{\mathbf{B}}\mathbf{\hat{l}.}
\end{equation}%
At the receivers, the modulo operation is applied, yielding%
\begin{equation}\label{SIC-receiver}
\begin{split}
\mathbf{y}\func{mod}A& =\mathbf{H}(\mathbf{Bx}-A\widetilde{\mathbf{B}}%
\mathbf{\hat{l}})+\mathbf{n}\text{ }\mod A \\
& =\mathbf{x-}A\mathbf{U\hat{l}}+\mathbf{n}\text{
}\mathrm{\mod}\text{
}A \\
& =\mathbf{x}+\mathbf{n}\mod A.
\end{split}%
\end{equation}

\subsubsection{ZF Precoding}

Let $\mathbf{u}=\mathbf{Bx/}A$. An estimate of $\mathbf{\hat{l}}$ is
found by ZF as follows%
\begin{eqnarray}
\mathbf{\hat{l}} \mathbf{=} \left\lceil \widetilde{\mathbf{B}}^{-1}\mathbf{u%
}\right\rfloor  \notag = \left\lceil
\frac{\mathbf{U}^{-1}\mathbf{x}}{A}\right\rfloor .
\end{eqnarray}%
The transmitted signal is given by%
\begin{eqnarray}
\mathbf{s} &=&\mathbf{Bx}-A\widetilde{\mathbf{B}}\mathbf{\hat{l}}  \notag \\
&=&\widetilde{\mathbf{B}}\left( \mathbf{U}^{-1}\mathbf{x}-A\left\lceil \frac{%
\mathbf{U}^{-1}\mathbf{x}}{A}\right\rfloor \right)  \notag \\
&=&\widetilde{\mathbf{B}}\left( \mathbf{U}^{-1}\mathbf{x}\text{ }\func{mod}%
A\right) .  \label{RF_Precoding_Equivalent}
\end{eqnarray}%
The second line of (\ref{RF_Precoding_Equivalent}) represents the
transmission scheme in \cite{Windpassinger04}, while the third line
corresponds to the transmission scheme in \cite{Taherzadeh07}.
Therefore, the schemes proposed in
\cite{Windpassinger04,Taherzadeh07} are equivalent. To the best of
our knowledge, this equivalence is not known in literature. At the
receivers, the modulo operation is applied, yielding the same as
(\ref{SIC-receiver}).

\subsection{Reduction Criteria}

To summarize, the purpose of approximate lattice precoding is to find a
sub-optimal solution $\mathbf{\hat{l}}$ that can reduce the norm $%
\left\Vert \mathbf{s}\right\Vert $. Withe lattice reduction, the
transmitted vector $\mathbf{s}$ falls into the fundamental
parallelepiped (for ZF) or the rectangle spanned by the Gram-Schmidt
vectors of the reduced basis $A\mathbf{B}$ (for SIC). In both cases,
the transmission power is proportional to the second moment over the
respective regions. Let $V=|\det \mathbf{B}|$ be the fundamental
volume of $\mathcal{L}(\mathbf{B})$, and $\mathcal{P}$ be its
fundamental parallelepiped. Let $\|\mathbf{B}\|^2$ be the Frobenius
norm of $\mathbf{B}$, and $\mathbf{\hat{B}}$ be the Gram-Schmidt
matrix for $\mathbf{B}$. Using a uniform-distribution approximation,
the transmission powers associated with the approximate lattice
precoders are respectively given by
\begin{equation}\label{}
P_{\text{ZF}} =
\frac{A^2}{V}\int_{\mathcal{P}}\|\mathbf{x}\|^2d\mathbf{x} =
\frac{A^2}{12}\sum_{i=1}^n{\|\mathbf{b}_i\|^2}
=\frac{A^2}{12}\|\mathbf{B}\|^2
\end{equation}
for ZF, and
\begin{equation}\label{}
P_{\text{SIC}} =
\frac{A^2}{12}\sum_{i=1}^n{\|\mathbf{\hat{b}}_i\|^2}=\frac{A^2}{12}\|\mathbf{\hat{B}}\|^2
\end{equation}
for SIC. Therefore, the objective of lattice reduction in this
application is to minimize the Frobenious norm of $\mathbf{B}$ or
$\mathbf{\hat{B}}$. However, it is computationally hard to exactly
accomplish this objective. Thus, the LLL algorithm is often used.

\section{Proximity Factors}

We want to understand the performance of approximate lattice
precoding. To do this, we compare the transmission powers with that
of sphere encoding, under the condition that they have the same
error performance, namely, (\ref{rcv}) and (\ref{SIC-receiver}) hold
at the receivers. This is a standard approach to calculating the
``coding gain"\footnote{In practice, power normalization is applied
at the transmitter in vector perturbation \cite{Hochwald05}, yet
such a scaling factor has no impact on the ``coding gain".}.

The transmission power of sphere precoding is given by
\begin{equation}
P_{\text{SP}} = \frac{A^2}{V}\int_{\mathcal{V}}\|\mathbf{x}\|^2
d\mathbf{x} \triangleq A^2 \sigma^2(\mathcal{V})
\end{equation}
where $\mathcal{V}$ denotes the Voronoi region of
$\mathcal{L}(\mathbf{B})$, and $\sigma^2(\mathcal{V})$ is the second
moment of $\mathcal{V}$. Then, the SNR gap is asymptotically given
by
\begin{equation}
\rho = \left\{
  \begin{array}{ll}
    \frac{\|\mathbf{B}\|^2}{12\sigma^2(\mathcal{V})}, & \hbox{for ZF;} \\
    \frac{\|\mathbf{\hat{B}}\|^2}{12\sigma^2(\mathcal{V})}, & \hbox{for SIC.}
  \end{array}
\right.
\end{equation}

Unfortunately, it is difficult to compute $\rho$, and we resort to
the proximity factors of LLL reduction-aided precoding, which
measure the worst-case loss in power efficiency relative to sphere
precoding. More formally, we
define the proximity factor as %
\begin{equation}
F_{P}\triangleq
\left\{
  \begin{array}{ll}
    \sup \frac{\Vert \mathbf{s}\Vert _{\text{ZF}}^{2}}{\Vert \mathbf{%
s}\Vert _{\text{SP}}^{2}}, & \hbox{for ZF;} \\
    \sup \frac{\Vert \mathbf{s}\Vert _{\text{SIC}}^{2}}{\Vert \mathbf{%
s}\Vert _{\text{SP}}^{2}}, & \hbox{for SIC.}
  \end{array}
\right.
\label{PF}
\end{equation}%
Obviously, $\rho \leq F_P$. This viewpoint implies that the precoding problem is $\eta $-CVP:%
\begin{equation}
{\Vert \mathbf{B}(\mathbf{x}-A\mathbf{\hat{l}}}{)\Vert }\leq \eta
\min_{\mathbf{l}\in \mathbb{Z}^{n}}{\Vert \mathbf{B}(\mathbf{x}-A\mathbf{l}%
)\Vert ,}
\end{equation}%
and consequently, $F_{P}\leq \eta ^{2}$. Babai derived the value of
$\eta$ \cite{Babai} in the case of $\delta =3/4$. In what follows,
we will derive the bounds in the general case. Let
$\alpha=1/(\delta-1/4)$.

\begin{lem}
\label{PF_P_SIC} If the lattice basis is LLL-reduced, then SIC
solves $\eta$-CVP for $\eta=\eta_n= \alpha ^{n/2}/\sqrt{ \alpha
-1}$.
\end{lem}

\begin{proof}
Let $%
\mathbf{B}$ be a LLL reduced basis and $\mathbf{B}=\mathbf{\hat{B}}\mu ^{T}$
be the GSO of the basis $\mathbf{B}$. Given a vector $\mathbf{y\in }$ $%
\mathbb{R}^{m}$, we write $\mathbf{y}$ as a linear combination of the GS
vectors $\mathbf{y=}\sum_{i=1}^{n}\beta _{i}\mathbf{\hat{b}}_{i}$. Let $%
\mathbf{u=}\sum_{i=1}^{n}p_{i}\mathbf{\hat{b}}_{i}$ be the nearest neighbor
of $\mathbf{y}$ in $\mathcal{L}\left( \mathbf{B}\right) $. Let $\theta $ be
the integer nearest to $\beta _{n}$ and $\mathbf{y}^{\prime }\mathbf{=}%
\sum_{i=1}^{n-1}\beta _{i}\mathbf{\hat{b}}_{i}+\theta
\mathbf{\hat{b}}_{n}$, and $\mathbf{v=}\theta \mathbf{b}_{n}$. For
$n=1$, SIC can find the closest vector $\mathbf{u}$. For $n\geq 2$,
we
have%
\begin{eqnarray}
\left\Vert \mathbf{y}-\mathbf{y}^{\prime }\right\Vert =\left\vert
\theta -\beta _{n}\right\vert \Vert \mathbf{\hat{b}}_{n}\Vert \leq
\frac{\Vert \mathbf{\hat{b}}_{n}\Vert }{2},  \label{aCVP_NRP_1}
\end{eqnarray}%
and%
\begin{eqnarray}
\left\Vert \mathbf{y}-\mathbf{u}\right\Vert &=&\sqrt{\sum_{i=1}^{n}\left%
\vert \beta _{i}-p_{i}\right\vert ^{2}\Vert \mathbf{\hat{b}}_{i}\Vert }
\notag \\
&\geq &\left\vert \beta _{n}-p_{n}\right\vert \Vert \mathbf{\hat{b}}_{n}\Vert
\notag \\
&\geq &\left\vert \beta _{n}-\theta \right\vert \Vert \mathbf{\hat{b}}%
_{n}\Vert  \notag \\
&=&\left\Vert \mathbf{y}-\mathbf{y}^{\prime }\right\Vert .
\label{aCVP_NRP_2}
\end{eqnarray}%
Let $\mathbf{w}$ be the estimate of $\mathbf{u}$ found by SIC. From (\ref{aCVP_NRP_1}), we obtain%
\begin{equation}
\left\Vert \mathbf{y}-\mathbf{w}\right\Vert ^{2}\leq \frac{1}{4}%
\sum_{i=1}^{n}\left\Vert \mathbf{\hat{b}}_{i}\right\Vert ^{2}.
\end{equation}%
According to (\ref{R-LLL-2}), we have%
\begin{equation}
\left\Vert \mathbf{y}-\mathbf{w}\right\Vert \leq \frac{1}{2}\sqrt{\frac{%
\alpha ^{n}-1}{\alpha -1}}\Vert \mathbf{\hat{b}}_{n}\Vert .
\label{aCVP_case2}
\end{equation}%
If $p_{n}=\theta $, then%
\begin{equation}
\left\Vert \mathbf{y}^{\prime }-\mathbf{w}\right\Vert \leq \eta
_{n-1}\left\Vert \mathbf{y}^{\prime }-\mathbf{u}\right\Vert \leq
\eta _{n-1}\left\Vert \mathbf{y}-\mathbf{u}\right\Vert .
\label{aCVP_NRP_D1}
\end{equation}%
By (\ref{aCVP_NRP_2}) and (\ref{aCVP_NRP_D1}),
\begin{eqnarray*}
\left\Vert \mathbf{y}-\mathbf{w}\right\Vert &=&\left( \left\Vert \mathbf{y}-%
\mathbf{y}^{\prime }\right\Vert ^{2}+\left\Vert \mathbf{y}^{\prime }-\mathbf{%
w}\right\Vert ^{2}\right) ^{1/2} \\
&\leq &(1+\eta _{n-1}^{2})^{1/2}\left\Vert
\mathbf{y}-\mathbf{u}\right\Vert
\\
&<&\eta _{n}\left\Vert \mathbf{y}-\mathbf{u}\right\Vert .
\end{eqnarray*}%
If $p_{n}\neq \theta $, then%
\begin{equation}
\left\Vert \mathbf{y}-\mathbf{u}\right\Vert \geq \frac{\Vert \mathbf{\hat{b}}%
_{n}\Vert }{2}.
\end{equation}%
Combining this inequality with (\ref{aCVP_case2}), we obtain%
\begin{eqnarray}
\left\Vert \mathbf{y}-\mathbf{w}\right\Vert &\leq &\sqrt{\frac{\alpha ^{n}-1%
}{\alpha -1}}\left\Vert \mathbf{y}-\mathbf{u}\right\Vert  \notag \\
&<&\frac{\alpha ^{n/2}}{\sqrt{\alpha -1}}\left\Vert \mathbf{y}-\mathbf{u}%
\right\Vert .
\end{eqnarray}
\end{proof}

For $\alpha =2$, $\eta _{n}=\alpha ^{n/2}/\sqrt{\alpha -1}$ reduces
to Babai's upper bound $2^{n/2}$.

\begin{lem}
\label{PF_P_ZF} If the lattice basis is LLL-reduced, then ZF solves
$\eta$-CVP with $\eta=\eta _{n}=1+2n\left(
3\sqrt{\alpha}/2\right)^{n-1}$.
\end{lem}

\begin{proof}
\label{Prox_pre_zf}
Let $\mathbf{B}$ be a LLL reduced basis. Let
$\theta _{i}$ be the angle between $\mathbf{b}_{i}$ and the linear
space $\mathcal{S}\left( \left[
\mathbf{b}_{1}\text{,...,}\mathbf{b}_{i-1}\text{,}\mathbf{b}_{i+1}\text{,...,%
}\mathbf{b}_{n}\right] \right) $ spanned by the other $n-1$ basis
vectors. Recall the following bound \cite{LingIT07}
\begin{eqnarray}
\sin \theta _{i} &\geq & \left( \frac{2}{3\sqrt{\alpha}}\right)
^{n-1}. \label{LLL-ZF-RB}
\end{eqnarray}%
Since%
\begin{equation}
\sin \theta _{i}=\min_{\mathbf{m\in }\mathcal{S}}\frac{\left\Vert \mathbf{m-b%
}_{i}\right\Vert }{\left\Vert \mathbf{b}_{i}\right\Vert },
\label{Sin_difi}
\end{equation}%
we have
\begin{equation}\label{m-b-bound}
\left\Vert \mathbf{m-b}_{i}\right\Vert \geq \left(
\frac{2}{3\sqrt{\alpha}}\right) ^{n-1}\left\Vert
\mathbf{b}_{i}\right\Vert, \quad \forall \mathbf{m} \in \mathcal{S}.
\end{equation}

Let $\mathbf{w}$ be the lattice point found by ZF. Then%
\begin{equation}
\mathbf{w}-\mathbf{y=}\sum_{i=1}^{n}\beta _{i}\mathbf{b}_{i},
\end{equation}%
where $\left\vert \beta _{i}\right\vert \leq 1/2$, for $1\leq i\leq n$. Let $%
\mathbf{u}$ be the nearest neighbor of $\mathbf{y}$ in $\mathcal{L}\left(
\mathbf{B}\right) $. We may write%
\begin{equation}
\mathbf{u-w=}\sum_{i=1}^{n}\phi _{i}\mathbf{b}_{i},
\end{equation}%
where $\phi _{i}\in \mathbb{Z}$. We assume $\mathbf{u\neq w}$. Let $%
\left\Vert \phi _{k}\mathbf{b}_{k}\right\Vert =\max_{i}\left\Vert \phi _{i}%
\mathbf{b}_{i}\right\Vert $. Then%
\begin{equation}
\left\Vert \mathbf{u-w}\right\Vert \leq n\left\Vert \phi _{k}\mathbf{b}%
_{k}\right\Vert .  \label{aCVP_ZF_E1}
\end{equation}%
Meanwhile,%
\begin{eqnarray*}
\mathbf{u-y} &\mathbf{=}&\left( \mathbf{u-w}\right) +\left( \mathbf{w-y}%
\right) \\
&=&\left( \phi _{k}\mathbf{+}\beta _{k}\right) \left( \mathbf{b}%
_{k}-\mathbf{m}\right) ,
\end{eqnarray*}%
where%
\begin{equation*}
\mathbf{m}=-\frac{1}{\phi _{k}\mathbf{+}\beta
_{k}}\sum_{j\mathbf{\neq }k}\left( \phi _{j}\mathbf{+}\beta
_{j}\right) \mathbf{b}_{j}.
\end{equation*}%
By (\ref{m-b-bound}) and $\left\vert \beta _{i}\right\vert \leq 1/2$, we have%
\begin{eqnarray}
\left\Vert \mathbf{u-y}\right\Vert &\geq &\frac{\left\vert \phi
_{k}\right\vert }{2\left( 3\sqrt{\alpha}/2\right) ^{n-1}}\left\Vert
\mathbf{b}_{k}\right\Vert . \label{aCVP_ZF_E2}
\end{eqnarray}%
Combining (\ref{aCVP_ZF_E1}) and (\ref{aCVP_ZF_E2}), we have%
\begin{eqnarray}
\left\Vert \mathbf{u-w}\right\Vert &\leq &n\left\Vert \phi _{k}\mathbf{b}%
_{k}\right\Vert  \notag \\
&\leq &2n\left( 3\sqrt{\alpha}/2\right) ^{n-1} \left\Vert
\mathbf{u-y}\right\Vert .
\end{eqnarray}%
It is easy to see that%
\begin{eqnarray}
\left\Vert \mathbf{y-w}\right\Vert &\leq &\left\Vert \mathbf{y-u}\right\Vert
+\left\Vert \mathbf{u-w}\right\Vert  \notag \\
&\leq &\left( 1+2n\left( 3\sqrt{\alpha}/2\right) ^{n-1}\right)
\left\Vert \mathbf{y-u}\right\Vert.
\end{eqnarray}
\end{proof}

For $\alpha =2$, we have $\eta _{n}=1+2n\left( 9/2\right) ^{\left(
n-1\right) /2}$.

From the two lemmas, we have the following theorem for the proximity
factors:

\begin{theorem}
According to Lemmas \ref{PF_P_SIC} and \ref{PF_P_ZF}, we have%
\begin{eqnarray}
F_{P,\text{SIC}} &\leq&\frac{\alpha ^{n}}{\alpha -1} \label{PF_SIC}
\end{eqnarray}%
and%
\begin{eqnarray}
F_{P,\text{ZF}} &\leq &\left( 1+2n\left( 3\sqrt{\alpha}/2\right)
^{n-1}\right) ^{2}.
\end{eqnarray}%
\end{theorem}

These results show that the worst-case loss in power efficiency of
approximate lattice precoders is bounded above by a function of the
dimension of the lattice alone.

\section{Discussion}

Our main contribution in this paper was to view the LR-aided
precoding problem as $\eta $-CVP, compared to the viewpoint of
bounded distance decoding for LR-aided decoding \cite{LingIT07}.
This viewpoint allowed us to derive the proximity factors, which
measure the worst-cased bound for approximate lattice precoding.
Since the underlying lattice is infinite, this analysis is rigorous,
and it follows that LR-aided precoding also achieves full diversity.
The derived bounds may not be tight, but nonetheless give more
insights. Improving the bounds is the future work.

\section*{ACKNOWLEDGMENT}

The authors are grateful to the reviewers for their helpful
comments. The work of Xiaofu Wu was supported by the National
Science Foundation of China under Grants 60972060, 61032004 and the
National Key S\&T Project under Grant 2010ZX03003-003-01.

\balance

{\small
\bibliographystyle{IEEEtran}
\bibliography{IEEEabrv,LINGBIB}
}

\vfill

\end{document}